\documentclass[twoside]{IEEEtran}
\usepackage{vkmacros}
\usepackage[left=.75in,right=.75in,top=1in,bottom=.75in,%
            footskip=.25in]{geometry}

\usepackage{ifthen}
\graphicspath{{figures/}{bios/}}

\newif\ifshowtodo
\showtodotrue

\newif\ifver
\verfalse 
\long\def\ver#1{\ifver{\color{blue!80!black}#1}\else{#1}\fi}

\newcommand{\VersionLength}{long}
\providecommand{\verlong}{\ifthenelse{\equal{\VersionLength}{long}}}
\newcommand{\VersionCols}{double}
\providecommand{\dcol}{\ifthenelse{\equal{\VersionCols}{double}}}

\title{Exact minimum number of bits \\
to stabilize a linear system}
\author{Victoria Kostina, Yuval Peres, Gireeja Ranade, Mark Sellke
\thanks{
V.~Kostina (\href{mailto:vkostina@caltech.edu}{vkostina@caltech.edu}) is with California Institute of Technology, Pasadena, CA.
Y.~Peres (\href{mailto:yuval@yuvalperes.com}{yuval@yuvalperes.com}) is an independent researcher.
G.~Ranade (\href{mailto:gireeja@eecs.berkeley.edu}{ranade@eecs.berkeley.edu}) is with the University of California, Berkeley, CA.
M.~Sellke (\href{mailto:msellke@stanford.edu}{msellke@stanford.edu}) is with Stanford University, CA. This work was supported in part by the National Science Foundation (NSF)
under Grant CCF-1751356, and by the Simons Institute for the Theory of Computing. Research of Y.~Peres was partially supported by NSF grant DMS-1900008. G.~Ranade acknowledges the Siebel Energy Institute Seed Funding.
}
}
\date{}							
\begin{document}
\maketitle

\begin{abstract}
We consider an unstable scalar linear stochastic system, $X_{n+1}=a X_n + Z_n - U_n$, where $a \geq 1$ is the system gain, $Z_n$'s are independent random variables with bounded $\alpha$-th moments, and $U_n$'s are the control actions that are chosen by a controller who receives a single element of a finite set $\{1, \ldots, M\}$ as its only information about system state $X_i$.  We show new proofs that $M > a$ is necessary and sufficient for $\beta$-moment stability, for any $\beta < \alpha$. Our achievable scheme is a  uniform quantizer of the zoom-in / zoom-out type \ver{that codes over multiple time instants for data rate efficiency; the controller uses its memory of the past to correctly interpret the received bits.} We analyze its performance using probabilistic arguments. We \ver{show a simple proof of} a matching converse using information-theoretic techniques. Our results generalize to vector systems, to systems with dependent Gaussian noise, and to the scenario in which a small fraction of transmitted messages is lost.
\end{abstract}
\begin{IEEEkeywords}
Linear stochastic control, source coding, data rate theorem. 
\end{IEEEkeywords}

\section{Introduction}
We study the tradeoff between stabilizability of a linear stochastic system and the coarseness of the quantizer used to represent the state.  The evolution of the system is described by
\begin{align}
 X_{n+1} = a X_n  + Z_n - U_n, \label{eq:systemscalar}
\end{align}
where constant $a \geq 1$; $X_1$ and $Z_1, Z_2, \ldots$ are independent random variables with bounded $\alpha$-th moments, and $U_n$ is the control action chosen based on the history of quantized observations.  More precisely, an \emph{$M$-bin causal quantizer-controller} for $X_1, X_2, \ldots$ is a sequence $\{\mathsf f_n, \mathsf g_n\}_{n = 1}^\infty$, where $\mathsf f_n \colon \mathbb R^n \mapsto [M]$ is the encoding (quantizing) function, and $\mathsf g_n \colon [M] \mapsto \mathbb R^n$ is the decoding (controlling) function, and $[M] \triangleq \{1, 2, \ldots, M\}$. At time $i$, the controller outputs
\begin{align}
 U_n = \mathsf g_n( \mathsf f_1(X_1), \mathsf f_2(X^2), \ldots, \mathsf f_n (X^n)).
\end{align}
The fundamental operational limit of quantized control of interest in this paper is the minimum number of quantization bins to achieve $\beta$-moment stability:
\begin{align}
M^\star_\beta \triangleq \min \bigg\{ M \colon &\exists\, M\text{-bin causal quantizer-controller } 
\dcol{\notag \\&}{}   \text{ s.t.~ } \lim \sup_n \E{|X_n|^\beta} < \infty  \bigg\},
\label{eq:Mstardef}
\end{align}
where  $0<\beta<\alpha$ is fixed. 

The main results of the paper are new proofs of the following achievability and converse theorems, whose various special cases have been previously shown in literature. 
\begin{thm}[achievability]
\label{thm:main}
Let $X_1, ~Z_n$ in \eqref{eq:systemscalar} be independent random variables with bounded $\alpha$-moments. Then for any $0<\beta<\alpha$
\begin{equation}
M^\star_\beta \leq \lfloor a\rfloor + 1. \label{eq:Mstar}
\end{equation}
\end{thm}

\begin{thm}[converse]
\label{thm:mainc}
Let $X_1$, $Z_n$ in \eqref{eq:systemscalar} be independent random variables. Let $h(X_1) > - \infty$, where $h(X) \triangleq - \int_{\mathbb R} f_{X}(x) \log f_{X}(x) dx $ is the differential entropy. Then, for all $\beta > 0$, 
\begin{equation}
M^\star_\beta \geq \lfloor a\rfloor + 1. \label{eq:Mstarc}
\end{equation}
\end{thm}




The first achievability results \cite{baillieul1999feedback,wong1999systems} focused on unstable scalar systems with bounded disturbances, i.e. $|Z_n| \leq B$ a.s., and showed that a simple uniform quantizer with the number of quantization bins in \eqref{eq:Mstar} stabilizes such systems. That corresponds to the special case $\alpha = \beta = \infty$. Nair and Evans \cite{nair2004stabilizability} showed that time-invariant fixed-rate quantizers are unable to attain bounded cost if the noise is unbounded \cite{nair2004stabilizability}, regardless of their rate. The reason is that since the noise is unbounded, over time, a large magnitude noise realization will inevitably be encountered, and the dynamic range of the quantizer will be exceeded by a large margin, not permitting recovery. This necessitates the use of adaptive quantizers of zooming type \cite{gersho1974training,kieffer1983type,brockett2000quantized}.  Such quantizers ``zoom out'' (i.e. expand their quantization intervals) when the system is far from the target and  ``zoom in'' when the system is close to the target. 
Nair and Evans~\cite{nair2004stabilizability} constructed such an adaptive fixed-length quantizer with nonuniform quantization levels and showed second-moment stability via a recursive bound on its mean-squared error, under the assumption that the system noise has bounded $2+\epsilon$ moment, for some $\epsilon > 0$.  Under the same assumption, Y\"uksel \cite{yuksel2010fixedrate} {(see \cite{johnston2014stochastic},\cite{yuksel2012drift} for  generalizations to vector systems and to $\beta = 1, 2, \ldots$)} showed second-moment stabilizability using a uniform scalar quantizer that enters its zoom-out mode whenever its input falls outside its dynamic range. When applied to encode each $k$-th system state over the following $k$ time instances, the schemes in \cite{nair2004stabilizability,yuksel2010fixedrate} attain \eqref{eq:Mstar} for a large enough~$k$. See also \cite{sabag2020stabilizing}, which explores the use of constrained quantizers to encode the overflow event over multiple time instances. 

The converse in the special case of $\beta = 2$ was proved in \cite{nair2004stabilizability}, where it was shown that it is impossible to achieve second moment stability in the system in \eqref{eq:systemscalar} using a quantizer-controller with the number of bins $< \lfloor a\rfloor + 1$. This implies the validity of  \thmref{thm:mainc} for $\beta \geq 2$. Variants of the necessity result in \thmref{thm:mainc} are known for vector systems with bounded disturbances \cite{tatikonda2004control} and model uncertainty \cite{martins2006feedback}; for noiseless vector systems under  different stability criteria \cite{yuksel2006minimum}; for vector linear stochastic systems under second moment constraint stabilized over rate-constrained noiseless \cite{nair2004stabilizability} and packet-drop \cite{minero2009data} channels;  for vector linear stochastic systems stabilized in probability over noisy channels \cite{matveev2008state}; and for nonlinear systems with additive noise stabilized in probability over noisy channels \cite[Th. 3.1]{yuksel2016nonlinear}.

In this paper, we construct a new zoom-in zoom-out scheme that most of the time operates as if the noise were bounded, and relies on a periodic magnitude test to determine whether the state has left the quantized region. Similar to an application of the schemes in \cite{nair2004stabilizability,yuksel2010fixedrate} to an undersampled system with the transmission of a codeword over multiple time slots mentioned above, our strategy uses coding over multiple time instants, and the controller uses its memory of the past to correctly interpret the received bits. While the controller in the above modification of the known schemes is almost always silent, producing a large signal once in $k$ time instances, our controller is almost always active, producing a control signal optimized for bounded noise. Thus it introduces less delay. If the periodic magnitude test is failed, the quantizer-controller enters the zoom-out mode, which is essentially the same as in \cite{yuksel2010fixedrate}: the controller looks for the $X_n$ in exponentially larger intervals until it is located, at which point it returns to the zoom-in mode. We provide an elementary analysis of our scheme with an explicit bound on $k$ leading to \thmref{thm:main}.

We also present a short proof of the converse result in \thmref{thm:mainc} that uses information-theoretic arguments. We also provide an elementary converse proof  for stabilizability in probability that is tight for non-integer $a$.
 
In \secref{sec:achievable}, we describe our achievable scheme and give its analysis. In \secref{sec:converse}, we give a proof of the converse in \thmref{thm:mainc}. Our results generalize to  constant-length time delays, to control over communication channels that drop a small fraction of packets, to systems with dependent Gaussian noise, and to vector systems. These extensions are presented in \secref{section:generalizations}. This is a full version of the conference paper \cite{kostina2018exact}.  This paper presents full proofs (the proofs in \cite{kostina2018exact} are either omitted or replaced with proof outlines), and a much more comprehensive discussion of our results and their relationship to past and future research.  In addition, \thmref{thm:weak}, which presents a converse using an elementary probabilistic argument, is not contained in \cite{kostina2018exact}.

\section{Achievable scheme}
\label{sec:achievable}
\subsection{The idea}
Here we explain the idea of our achievable scheme. For readability we focus on the case $a\in [1,2)$ and show that the system can be controlled with 1 bit. In this case we will be able to restrict to two types of tests, a \emph{sign test} and a \emph{magnitude test} (see \figref{fig:quantizer}), which simplifies our procedure.   The straightforward extension to an arbitrary $a \geq 1$, in which the sign test is replaced by a uniform quantizer, is found in \secref{sec:finer} below.

\begin{figure}
\centering
  \centering
\subfigure[Sign test]{\label{fig:a}\includegraphics[width=60mm]{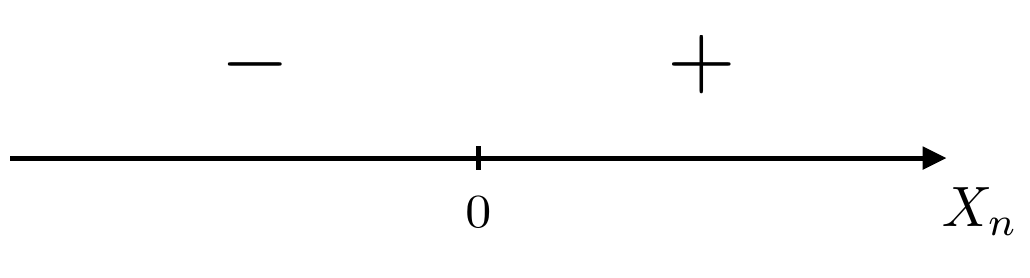}}
\subfigure[Magnitude test]{\label{fig:b}\includegraphics[width=60mm]{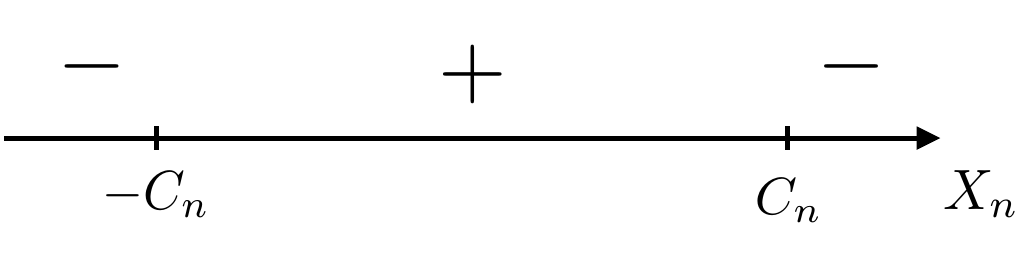}}
\caption{The binary quantizer uses two kinds of tests on a schedule determined by the previous $\pm$'s to produce the next $+$ or $-$.  }
\label{fig:quantizer}
\end{figure}

 In the case of bounded noise a uniform time-invariant  quantizer deterministically keeps $X_n$ bounded \cite{baillieul1999feedback,wong1999systems}. Indeed, when $|Z_n|\leq B$, $n = 1, 2, \ldots$ and $|X_1| \leq C_1$, if $C_1 \geq \frac {B}{1 - a/2}$ one can put  
\begin{align}
 C_2 \triangleq  (a/2) C_1 +B \leq C_1,
\end{align}
and putting further 
 $C_{n+1} \triangleq  (a/2) C_n+B$,
we obtain a monotonically decreasing to $\frac {B}{1 - a/2}$ sequence numbers $\{C_n\}_{n = 1}^\infty$. 
Setting
\begin{align}
 U_n = (a/2) C_n  \sign(X_n) \label{eq:Uibounded}
\end{align}
requires only 1 bit of knowledge about $X_n$ (i.e., its sign). If $|X_{n}| \leq C_n$ then 
\begin{align}
 |X_{n+1}|\leq (a/2) C_n +B = C_{n+1}, \label{eq:Xbounded}
\end{align}
and 
\begin{align}
 \limsup_{n \to \infty} |X_n| \leq \frac {B}{1 - a/2}.
\end{align}
Actually, this is the best achievable bound on the uncertainty about the location of $X_n$, as a simple volume-division argument shows \cite{tatikonda2004control,nair2007feedback}. 

When $Z_n$ merely have bounded $\alpha$-moments the above does not work because a single large value of $Z_n$ will cause the system to explode. However we can use the idea of the bounded case with the following modification. Most of the time, in \emph{normal}, or \emph{zoom-in}, mode, the controller assumes the $X_n$ are bounded by constants $C_n$ and forms the control actions according to the above procedure, but occasionally, on a schedule, the quantizer performs a \emph{magnitude test} and sends a bit  whose sole purpose is to inform the controller whether the  $X_n$ is staying within desired bounds. If the test is passed, the controller continues in the normal mode, and otherwise, it enters the \emph{emergency}, or \emph{zoom-out}, mode, whose purpose is to look for the $X_n$ in exponentially larger intervals until it is located, at which point it returns to the zoom-in mode while still occasionally checking for anomalies. We will show that all this can be accomplished with only 1 bit per controller action.



The intuition behind our scheme is the following.  At any given time, with high probability $X_n$ is not too large. Thus, the emergencies are rare, and when they do occur, the size of the uncertainty region tends to decrease exponentially. The zoom-in mode operates almost exactly as in the bounded case, except that we choose $B$ large enough to diminish the probability that the noise exceeds it. We now proceed to making these intuitions precise in \secref{sec:algo}.

\subsection{The Algorithm}
\label{sec:algo}
Here we describe the algorithm precisely and then prove that it works.  Specifically, we consider the setting of \thmref{thm:main} with $a \in [1, 2)$ and $Z_n$ with bounded $\alpha$-moments. 
We find $U_n$ - a function only of the sequence bits received from the quantizer - that achieves $\beta$-moment stability, for $0<\beta<\alpha$.

First we prepare some constants. 
We fix $B \geq 1$ large enough. We set the \emph{probing factor} $P=P(\alpha,\beta)$ - a large positive constant (how large will be explained below, but roughly $P$ blows up as $\beta\uparrow\alpha$).
Fix a small $\delta > 0$ and a large enough $k = k(a)$ so that
 \begin{align}
 ( a/2)^{k - 1} a &\leq 1 - 3 \delta \label{eq:k}.
\end{align} 
We proceed in ``rounds" of at least $k+1$ moves, $k$ moves in normal (zoom-in) mode and $k+1$'th move to test whether $X_n$ escaped the desired bounds. If that magnitude test comes back normal, the round ends; otherwise the controller enters the emergency (zoom-out) mode, whose duration is variable and which ends once the controller learns a new (larger) bound on $X_n$. 
In normal mode, we use the update rule in \eqref{eq:Uibounded}, where $C_n \geq B$ is positive. In the emergency mode, $U_n \equiv 0$ while $C_n$ grows exponentially. 
A precise description of the operation of the algorithm is given below. 

\begin{enumerate} 
\item  At the start of a round at time-step $m$, $|X_m| \leq C_m$, the controller is silent, $U_{m} = 0$, and $X_{m+1} = a X_m + Z_m$. Set 
\begin{align}
C_{m+1} = a C_m + B,  \label{eq:Cmplus1}
\end{align}
and for each $i\in\{2, \ldots ,k\}$,
\begin{align}
C_{m+i}&=\frac { a} 2 C_{m+i-1} + B  \label{eq:Cmplusi} \\
&= \left(a / 2\right)^{i-1} C_{m+1} + \frac{ 1 - ({ a}/2)^{i-1}} {1 - { a}/2} B. \label{eq:Cmplusirec}
\end{align}
In this normal mode operation, the quantizer sends a sequence of signs of $X_n$ (see \figref{fig:a}), while the controller applies the controls \eqref{eq:Uibounded} successively to $X_{m}, \ldots, X_{m+k - 1}$.  This normal mode operation will keep $X_{m+i}$ bounded by $C_{m + i}$
unless some $Z_{m+i}$ is atypically large.
\item The quantizer applies the magnitude test to check whether $|X_{m+k}| \leq C_{m+k}$ (see \figref{fig:b}). If $|X_{m+k}| \leq C_{m+k}$, we return to step 1. If $|X_{m+k}| > C_{m+k}$, this means some $Z_{m+i}$ was abnormally large; the system has blown up and we must do damage control. In this case we enter \emph{emergency (zoom-out) mode} in Step 3 below.

\item In emergency mode, we repeatedly perform silent ($U_{m+k+j} \equiv 0$) magnitude tests via 
\begin{equation}
 C_{m + k +j}= P\, C_{ m + k + j-1} = P^{j} C_{m+k}  \quad j \geq 0
 \label{eq:emergency}
\end{equation}
until the first time $\tau$ that the magnitude test is passed, i.e. 
\begin{equation}
\tau \triangleq \inf \left\{ j \geq 0 \colon |X_{m+k + j}| \leq C_{m+k+j} \right\}. \label{eq:tau}
\end{equation}
We then set $m \leftarrow m + k + \tau$ and return to Step 1.
\end{enumerate}
The controller is silent at the start of a round because it does not know the sign of $X_m$. Each round thus includes one silent step at the start, and $\tau \geq 0$ silent steps of the emergency mode.

\subsection{Overview of the Analysis} 
\label{sec:overview}
We analyze the result of each round. At the start of each round $m$ we know that $X_m$ is contained within interval $[-C_m, C_m]$. We will show that when $C_m$ is large, the uncertainty interval tends to decrease by a constant factor each round. 

At the start of the round,  $|X_{m}| \leq C_{m}$. Assume that for each $i\in \{ 0,1,\ldots,k\}$, we have 
\begin{equation}
\label{eqn:noiseisfine}
|Z_{m+i}|\leq B.
\end{equation}
and thus
\begin{align}
 |X_{m+i}| &\leq C_{m+i}. \label{eq:Xmplusibound}
\end{align}
In particular, applying \eqref{eq:k}, \eqref{eq:Cmplus1} and \eqref{eq:Cmplusi}, we bound the state at the end of the round as
\begin{align}
 |X_{m+k}| &\leq C_{m+k} \label{eq:Xwithinbound}\\
 &\leq (1 - 3 \delta)\, C_{m} + \frac{ B } {1 - { a}/2} \label{eq:Xmpluskbound},
\end{align}
which means that $C_{m+k} \leq C_m$, provided that $C_m \geq \frac{B}{3 \delta (1 - a/2)}$. Thus, even starting with the silent step we have successfully decreased $C_{m}$, provided that it was large enough.

What if \eqref{eqn:noiseisfine} fails to hold? Because the $Z_i$ have bounded $\alpha$-moments, by the union bound and Markov's inequality, the chance \eqref{eqn:noiseisfine} fails is at most 
\begin{equation}
 \Prob{ \cup_{i = 0}^{k} \left\{ {|Z_{m+i}| > B} \right\}} \leq (k+1)\, \E{|Z|^\alpha} B^{-\alpha}.\label{eq:Zimarkov}
\end{equation} 

 In this case, we show that we can control the blow-up to avert a catastrophe. 
 Recall that in emergency mode our procedure will take exponentially growing $C_n$ (see \eqref{eq:emergency}) so that we will soon observe that $|X_n| \leq C_n$. The controller then exits emergency mode and returns to the normal mode, starting a new round at time step $n$. Using boundedness of $\alpha$-moments of $Z_i$, we will show in \secref{sec:precise} below that the chance that on step $n = m+k+j$ this \emph{fails} is exponentially small in $j$. We will see that in each round starting at $X_m \in [-C_m, C_m]$, there is a high chance to shrink the magnitude of the state and a small chance to grow larger. In the next section we explain how to obtain precise moment control.


\subsection{Precise Analysis}
\label{sec:precise}

Here we give details of the analysis outlined in \secref{sec:overview}, demonstrating that when the $Z_n$ are i.i.d. with bounded $\alpha$-moments, our strategy in \secref{sec:algo} yields 
\begin{align}
 \limsup_n \mathbb E[|X_n|^{\beta}]<\infty \label{eq:goal}
\end{align}
for all $0<\beta<\alpha$. 

The following tools will be instrumental in controlling the tails of the accumulated noise.

\begin{prop}
\label{prop:moments}
If the random variable $Z$ has finite $\alpha$-moment, then 
\begin{align}
 t^{\alpha}\mathbb P[|Z|>t] \label{eq:Ztail}
\end{align}
are bounded in $t$. Conversely, if \eqref{eq:Ztail} are bounded in $t$ then $Z$ has a finite $\beta$-moment for any $0<\beta<\alpha$.
\end{prop}

\begin{proof}
The first part is the Markov inequality. The second is a standard use of the tail-sum formula.
\end{proof}

\begin{lemma}

\label{lemma:moments}

Suppose $a>1$ is fixed and $Z_i$ are (arbitrarily coupled) random variables with uniformly bounded absolute $\alpha$ moments. Then the random variables
\begin{equation}
 \tilde Z_{j} \triangleq \sum_{i = 0}^{j} a^{-i} Z_{i}
\end{equation}
also have uniformly bounded absolute $\alpha$-moments.
\end{lemma}

\begin{proof}
It is easy to see that for any $\alpha > 0,~\varepsilon>0$ there is $c = c_{\alpha,\varepsilon}$ such that for all 

\begin{equation}
 (x+y)^{\alpha}\leq c_{\alpha,\varepsilon}x^{\alpha}+(1+\varepsilon)y^{\alpha} \label{eq:calphaeps}
\end{equation}
holds for all $x, y \geq 0$. Indeed, to see this, assume without loss of generality that $x=1$, and note that when $y$ is sufficiently large we already have 
\begin{equation}
 (1+y)^{\alpha}\leq (1+\varepsilon)y^{\alpha}. \label{eq:calphaepsa}
\end{equation}
The set of $y$ for which \eqref{eq:calphaepsa} does not hold is bounded, hence so is the value of $(1+y)^{\alpha}$; take $c$ to be an upper bound for this expression. The equation will now hold for any value of $y$. 

Applying \eqref{eq:calphaeps} repeatedly yields 
\begin{align}
|\tilde Z_k|^\alpha &\leq \dcol{\\& \notag}{} c|Z_0|^{\alpha} + c \sum_{i = 1}^{k-1} (1 + \epsilon)^{i} a^{-\alpha i} |Z_i|^\alpha + (1+\varepsilon)^ka^{-\alpha k} |Z_k|^\alpha.
\end{align}
Since $\E{|Z_i|^\alpha}$ are uniformly bounded and for $1+\varepsilon<a^\alpha$ the geometric series  $\sum_{i = 1}^{j-1} (1 + \epsilon)^{i} a^{-\alpha i}$ converges, $\mathbb E[|\tilde Z_j|^{\alpha}]$ is bounded uniformly in $j$, as desired.
\end{proof}

\begin{remark}
The mild assumptions of  Lemma~\ref{lemma:moments} will make it easy to generalize our results to dependent noise in \secref{sec:dependent} below.
\end{remark}

The bound in \lemref{lemma:max} below considers the evolution of the system over $k + 1 + \tau$ steps, where $\tau$ \eqref{eq:tau} determines the end of the round. Note that $\tau$ is a stopping time of the filtration generated by $\{X_n\}$.  

\begin{lemma} 
Fix $B, P > 0$ and consider our algorithm described in \secref{sec:algo} with these parameters. Suppose that time-step $m$ is the start of a round, so that the round ends on time-step $m+k+\tau$. For all $1 < a < 2$ and for all $0 \leq j \leq \tau$, it holds that 
\begin{align}
  \label{eq:max}
\dcol{&~}{}
\max\left\{|X_{m+1}|, \ldots, |X_{m+k+j}|, C_{m+k+j}  \right\}
\dcol{\\ \notag}{}
  \leq &~ P a^{k+j}    \left( 2 C_m+  \frac{a\, B}{(2 - a)(a-1)} + \sum_{\ell = 0}^{k+ j - 1} a^{- \ell - 1 } |Z_{m + \ell} | \right),
\end{align}
\label{lemma:max}
\end{lemma}
\begin{proof}
Appendix. 
\end{proof}

\begin{proof}[Proof of \thmref{thm:main} for the case $a \in [1, 2)$]
To avoid a special treatment of the case $a = 1$, we assume that $a > 1$. This is without loss because showing stability for $a$ implies stability for all $a^\prime \leq a$.
First we prepare some constants. Recall the choices of $k$ and $\delta$ in \eqref{eq:k}.  
\begin{itemize}
\item Fix $\Delta<\alpha-\beta$ an arbitrary fixed constant, e.g. $\Delta=\frac{\alpha-\beta}{3}$, so that 
\begin{equation}
\beta = \alpha - 3 \Delta. \label{eq:betaalpha}
\end{equation}

\item Fix $P$ large enough so that
\begin{align}
 P/a \geq \max\left\{ \left(\frac{a}{1 - \delta}\right)^{\alpha - \Delta},~ 2^{k}, ~  \frac{a^{k+1}}{2(a - 1)} \right\} \label{eq:Pchoice}.
\end{align}
\end{itemize}

Suppose that time-step $m$ is the start of a round, so that the round ends on time-step $m+k+\tau$, with stopping time $\tau=0$ usually.

We define a modified sequence\footnote{\ver{$\tilde X_n$ serves as a Lyapunov function that stochastically
controls the growth of the state process $X_n$. See \cite[Th. 2.1]{yuksel2012drift}, \cite{yuksel2013stochastic} for a general approach to proving stability using Lyapunov functions for general Markov chains. } }
$\tilde X_n$ through, for $1 \leq i \leq k+\tau$, 
\begin{align}
\label{eq:Xtilde}
\tilde X_{m+i} \triangleq&~ \left( \frac 1 {1 - \delta}\right)^{ \tau-|i - k|_+}
\dcol{\\\notag&~}{} 
\max\left\{|X_{m+k}|, \ldots, |X_{m+k+\tau}|, C_{m+k+\tau}  \right\},
\end{align} 
where $|\cdot |_+ \triangleq \max\{0, \cdot\}$.
Clearly this definition ensures that 
\begin{align}
|X_{m+k+j}|  \leq \tilde X_{m+k+j} \quad  0 \leq j \leq \tau.
 \label{eq:bound}
\end{align}
Furthermore, for all $1 \leq i \leq k - 1$, there exists universal constants $K_1, K_2, K_3$ that depend on $a$, $k$ and $B$ such that (Appendix~\ref{sec:proofnormalbound})
\begin{equation}
\E{|X_{m+i}|^\beta} \leq K_1\, \E{ \tilde X_{m+k}^\beta} + K_2\, \E{  \tilde X_{m}^\beta} + K_3 \label{eq:normalbound}.
\end{equation}
Inequalities \eqref{eq:bound} and \eqref{eq:normalbound} together mean that to establish \eqref{eq:goal}, it is sufficient to prove
\begin{equation}
 \limsup_n \mathbb E[\tilde X_{n}^{\beta}] < \infty \label{eq:tildebound}.
\end{equation}
The rest of the proof is focused in establishing \eqref{eq:tildebound}.

By definition \eqref{eq:Xtilde},  
\begin{align}
 \tilde X_{m+i} \leq \tilde X_{m+1} \quad i = 2, \ldots, k + \tau, \label{eq:onlyfirstmatters}
\end{align}
with equality for $i \leq k$.

We will show that
\begin{align}\label{eq:momentdecrease}
\mathbb E[\tilde X_{m+1}^{\beta}] &\leq (1-\delta)^{\beta}\mathbb E[\tilde X_{m}^{\beta}]+K,
\end{align} 
where $K = K(P, k, \delta)$ is a constant that may depend on $P, k, \delta$ (but is independent of $m$). Together, inequalities \eqref{eq:onlyfirstmatters}  and \eqref{eq:momentdecrease} ensure that $\limsup_n \mathbb E[\tilde X_{n}^{\beta}]$ is bounded above by $\frac {K}{1 - (1-\delta)^{\beta}}$.

The intuition behind the definition for $\tilde X_n$ is as follows. We want to construct a dominating sequence $\tilde X_n$ with the expected decrease property in \eqref{eq:momentdecrease}. During emergency mode, the original sequence $X_n$ may increase on average during rounds. The sequence $\tilde X_n$ in \eqref{eq:Xtilde} takes the potential increase during each round up front, achieving the desired expected decrease property. We will see that $P$ in \eqref{eq:Pchoice} is chosen so that the constant-factor decrease of the system is preserved when switching between rounds.

To show \eqref{eq:momentdecrease}, we define the filtration $\mathcal F_n$ as follows: $\mathcal F_n$ is the $\sigma$-algebra generated by the sequences $Z_1, Z_2, \ldots, Z_{n-1}$ and $\tilde X_1, \tilde X_2, \ldots, \tilde X_n$. Unless $n$ is the end of a round, knowledge of $\tilde X_n$ involves a peek into the future, so $\mathcal F_n$ encompasses slightly more information than the naive notion of ``information up to time $n$". The inequality we will show, clearly stronger than \eqref{eq:momentdecrease}, is 
\begin{align}\label{eq:momentdecrease2}\mathbb E[\tilde X_{m+1}^{\beta} \mid \mathcal F_m ]\leq (1-\delta)^{\beta}\mathbb E[\tilde X_m^{\beta} \mid \mathcal F_m ]+K.
 \end{align}

Define
\begin{equation}
 Y_n \triangleq \frac{\tilde X_{n+1}}{\tilde X_{n}+ \frac{B}{(1 - a/2) (1 - 3 \delta)} }. \label{eq:Yn}
\end{equation}

We will show \eqref{eq:momentdecrease2} by the means of the following two statements, where $m$ is the transition between rounds: 
\begin{enumerate}[(a)]

 \item For sufficiently large $k$ and $P$ in \eqref{eq:k} and \eqref{eq:Pchoice}, respectively, it holds that 
 \footnote{Throughout this section, the implicit constants $\bigo{\cdot}$ may depend on $P, k, \delta$ (but are independent of $n$ and $B \geq 1$). }
\begin{equation}
 \mathbb P\left[Y_m \geq t | \mathcal F_{m} \right]  = \bigo{ t^{-(\alpha-\Delta)} },  \label{eq:Y_itail}
\end{equation} 
\item As $B \to \infty$, 
\begin{equation}
\Prob{Y_m \leq 1 - 3 \delta \mid \mathcal F_m } \to 1.\label{eq:Y_good}
\end{equation}
\end{enumerate}

We use \eqref{eq:Y_itail} and \eqref{eq:Y_good} to show \eqref{eq:momentdecrease2} as follows. First, observe that by \eqref{eq:Y_itail} and Proposition~\ref{prop:moments}, $\{Y_m | \mathcal F_m\}$ has bounded $\beta + \Delta$ - moment since we assumed \eqref{eq:betaalpha} when choosing $\Delta$. Furthermore, since the right side of \eqref{eq:Y_itail} is independent of $\mathcal F_m$,  the $\beta + \Delta$ - moment of $Y_m$ is bounded uniformly in $m$. Now, pick $p > 1$   so that $\beta p \leq \beta + \Delta$, and let $q$ satisfy $\frac 1 p + \frac 1 q = 1$. Write
\begin{align}
 \dcol{&~}{}
 \E{Y_m^\beta \mid \mathcal F_m} \dcol{\notag\\}{}
 \leq&~
 (1 - 3\delta)^\beta +
 \E{  Y_m^\beta\, \1{ Y_m >  1 - 3\delta } \mid \mathcal F_m } \\ 
 \leq&~    (1 - 3\delta)^\beta + \left(  \E{  Y_m^{\beta p} \mid \mathcal F_m } \right)^{\frac 1 p} \left( \Prob{ Y_m >  1 - 3\delta \mid \mathcal F_m } \right)^{\frac 1 q} \label{eq:expbounda}\\
 \to&~    (1 - 3\delta)^\beta, \quad B \to \infty 
 \label{eq:expboundb}, 
\end{align}
where \eqref{eq:expbounda} is by H\"older's inequality, and the second term in \eqref{eq:expbounda} vanishes  as $B \to \infty$ due to \eqref{eq:Y_good} and uniform boundedness of the $\beta + \Delta$ - moment of $\{Y_m \mid \mathcal F_m\}$. Note that convergence in \eqref{eq:expboundb} is uniform in $m$. It follows that for a large enough $B$ (how large depends on the values of $P, k, \delta$), 
\begin{equation}
 \E{Y_m^\beta \mid \mathcal F_m} \leq (1 - 2 \delta)^\beta. \label{eq:Yncondexp}
\end{equation}

Rewriting \eqref{eq:Yncondexp} using \eqref{eq:Yn} yields
\begin{align}
 \mathbb E[\tilde X_{m+1}^{\beta} \mid \mathcal F_m] &\leq (1 - 2 \delta)^\beta \left( \tilde X_m +\frac{B}{(1 - a/2) (1 - 3 \delta)}\right)^{\beta}\\
 &\leq (1 - \delta)^\beta \tilde X_m^{\beta}+K, \label{eq:momentsfinal}
\end{align}
where to write \eqref{eq:momentsfinal} we used \eqref{eq:calphaeps}. This establishes the inequality~\eqref{eq:momentdecrease2}.
 
To complete the proof of Theorem~\ref{thm:main}, it remains to establish \eqref{eq:Y_itail} and \eqref{eq:Y_good}.

To show \eqref{eq:Y_itail}, recall that the round ends at stopping time $m + k +  \tau$. 
Since the events $\{\tau = j\}$ are disjoint, we have
\begin{align}
 \mathbb P\left[Y_m\geq t | \mathcal F_m\right] &=  \sum_{j = 0}^\infty \mathbb P[Y_m\geq t, \tau = j | \mathcal F_m]  \dcol{\notag\\ &\phantom{=}}{} + \mathbb P[Y_m\geq t, \tau = \infty | \mathcal F_m] \label{eq:tailsum}
\end{align}

Note that since $m$ is the end of the previous round, $\mathcal F_m$ does not contain any information about the future. 

We estimate the probability of the event in $\mathbb P[Y_m\geq t, \tau = j | \mathcal F_m]$ in two ways, and use the better estimate on each term individually.

We express the system state at time $m + i$ in terms of the system state at time $m$:
 \begin{align} 
X_{m+i} &=
 a^{i} \left( X_m + \sum_{\ell = 0}^{i-1} a^{- \ell - 1} U_{m+ \ell}
+ \sum_{\ell = 0}^{i-1} a^{- \ell - 1 } Z_{m + \ell} \right). \label{eq:Xmj}
\end{align}

Using \eqref{eq:Uibounded}, \eqref{eq:Cmplus1}, \eqref{eq:Cmplusi} and recalling that $U_m = 0$, we can crudely bound the cumulative effect of controls on $X_{m+i}$ as
 \begin{align}
a^{i} \left| \sum_{\ell = 0}^{k-1} a^{ - \ell - 1} U_{m+\ell} \right|
&\leq a^{i} \left( a/2 \right)  \sum_{\ell = 1}^{\infty} a^{ - \ell - 1} 
 \dcol{\\ &\phantom{=}}{}
 \left( \left( a/2 \right)^{\ell- 1} C_{m+1} +   \frac{ 1 - ({ a}/2)^{\ell-1}} {1 - { a}/2} B \right) 
 \dcol{\notag}{}
 \\ 
&= 
a^{i} \left(C_{m}  + \frac B { a-1 } \right).
\label{eq:controlnormal} 
\end{align}

 Recalling the definitions of $\tilde X_n$, $Y_n$ in \eqref{eq:Xtilde}, \eqref{eq:Yn}, respectively, and invoking \lemref{lemma:max}, we see that if $\{Y_m \geq t, \tau = j\}$ holds, then
\begin{align}
 &~ t (1 - \delta)^{k + j - 1} \left( \tilde X_{m}+ \frac{B}{(1 - a/2) (1 - 3 \delta)} \right) \label{eq:eventmax}
  \dcol{\\ \leq&~ \notag}{\leq} 
  P a^{k+j}   \left( 2 C_m+  \frac{a B}{ (2-a) (a-1) } +  \sum_{\ell = 0}^{k+ j - 1} a^{- \ell - 1 } |Z_{m + \ell} |  \right).  
\end{align}
Noting that both $C_m$ and $\frac{a B}{ 2-a}$ are dominated by $\tilde X_{m}+ \frac{B}{(1 - a/2) (1 - 3 \delta)} \geq 1$, we can weaken \eqref{eq:eventmax} as 
\begin{align}
\dcol{\!\!\!}{}
 t (1 - \delta)^{k + j - 1} \leq P a^{k+j}     \left( 2 + \frac 1 {a-1} + \sum_{\ell = 0}^{k+ j - 1} a^{- \ell - 1 } |Z_{m + \ell} | \right). \label{eq:onehand}
\end{align}
Applying \lemref{lemma:moments} and \propref{prop:moments}, we deduce that the probability of the event in \eqref{eq:onehand} is 
\begin{align}
 O\left(  \left(\frac{a}{1-\delta}\right)^{\alpha j}t^{-\alpha}\right). \label{eq:prob2}
\end{align}

The bound in \eqref{eq:prob2} works well for small $j$ / large $t$. For large $j$ / small $t$, we observe that $\{Y_m \geq t, \tau = j\} \subseteq \{\tau \geq j\} $ and apply the following reasoning. The event $\{\tau \geq j\}$ means that the emergency did not end at time $j$; in other words, 
\begin{align}
 |X_{m+k+j-1}| &> C_{m+k+j-1} \label{eq:emergencydidnotend}
 \\
 &=  P^j \left( 2 \left(a/2 \right)^{k}  C_{m} + \frac{B}{1 - { a}/2} \right), \label{eq:emergencydidnotend3}
\end{align}
where to write \eqref{eq:emergencydidnotend3} we used \eqref{eq:Cmplus1}, \eqref{eq:Cmplusirec}, and \eqref{eq:emergency}. 
Substituting $i \leftarrow  k + j$ into \eqref{eq:Xmj} and recalling \eqref{eq:controlnormal} and $|X_m| \leq C_m$, we weaken \eqref{eq:emergencydidnotend}--\eqref{eq:emergencydidnotend3} as
\begin{align}
& \phantom{=} a^{k + j} \bigg(  2C_m + \frac{a B}{ (2-a) (a-1) } +  \sum_{\ell = 0}^{k+ j - 1} a^{- \ell - 1 } |Z_{m + \ell} |  \bigg)
\dcol{\notag\\ &}{}
>
 P^j \left( 2 \left(a/2 \right)^{k}  C_{m} + \frac{B}{1 - { a}/2} \right),
\end{align}
the event equivalent to 
\begin{align}
& (a/P)^j  \sum_{\ell = 0}^{k+ j - 1} a^{- \ell - 1 } |Z_{m + \ell} |  \geq   2 \left( (1/2)^{k}  -   (a/P)^{j} \right) C_{m} 
\dcol{\notag\\  & }{}
+ \left( (1/a)^k  -\frac {a (a/P)^j }{2(a - 1)} \right) \frac{B}{1 - a/2}.  \label{eq:otherhand}
\end{align}
Due to the choice of $P$ in \eqref{eq:Pchoice}, the coefficients in front of $C_m$ and $B$ in the right side of \eqref{eq:otherhand} are nonnegative for all $j \geq 1$. Bounding the probability of the event in \eqref{eq:otherhand} using \lemref{lemma:moments} and \propref{prop:moments}, we conclude that \footnote{\ver{Similar exponential bounds to the event $\Prob{\tau \geq j}$ are provided in  \cite[Lem. 5.2]{yuksel2012drift} and in \cite[Lem. 5.2]{yuksel2016nonlinear}}. }
\begin{equation}
\Prob{\tau \geq j} = O\left( \left( P/a\right)^{- j \alpha} \right). \label{eq:prob1}
\end{equation}
Furthermore, \eqref{eq:prob1} means that $\Prob{\tau = \infty} = 0$. Indeed, $1\{\tau = \infty\} =  \prod_{j = 0}^\infty \1{\tau \geq j } =\lim_{j \to \infty}\1{\tau \geq j}$ and by Fatou's lemma, 
\begin{align}
\Prob{\tau = \infty} \leq \lim_{j \to \infty} \Prob{ \tau \geq j } = 0, 
\end{align}
thus the corresponding term can be eliminated from \eqref{eq:tailsum}. 

Juxtaposing \eqref{eq:prob2} and  \eqref{eq:prob1}, we conclude that the probability $\mathbb P[Y_m\geq t, \tau = j | \mathcal F_m]$ is bounded by  
\begin{equation}
O\left( \min\left\{ \left(\frac{a}{1-\delta}\right)^{\alpha j}t^{-\alpha},  \left( P/a\right)^{- j \alpha}  \right\} \right). \label{eq:omin}
\end{equation}
Since \eqref{eq:Pchoice} ensures that $(P/a)^{\Delta}\geq \left(\frac{a}{1-\delta}\right)^{\alpha}$, 
 we weaken \eqref{eq:omin} as 
\begin{equation}
O\left( \left( P/a\right)^{ j \Delta}  \min\left\{ t^{-\alpha},  \left( P/a\right)^{- j \alpha}  \right\} \right). 
\end{equation}

Recall that we have fixed $t$ and are varying $j$; this upper bound peaks at $j$ such that $(P/a)^j=t$ at the value $t^{-(\alpha-\Delta)}$ and decays geometrically on each side at rates $(P/a)^{\Delta}$ and $(P/a)^{\alpha-\Delta}$. Hence the sum of all $\mathbb P[Y_m\geq t, \tau = j | \mathcal F_m]$ terms in \eqref{eq:tailsum} is bounded by the maximum up to a constant factor and therefore \eqref{eq:Y_itail} holds. 

To complete the proof of \thmref{thm:main}, it remains to establish \eqref{eq:Y_good}. By Markov's inequality \eqref{eq:Zimarkov}, with probability converging to $1$ as $B \to \infty$, all terms $Z_m, \ldots, Z_{m+k}$ are within $[-B, B]$, and $\tau = 0$. In such a case, applying \eqref{eq:Xmpluskbound} and recalling \eqref{eq:Xtilde}, we get
\begin{align}
\tilde X_{m+1} &= \max\{ |X_{m+k}|, C_{m+k} \} \\
&\leq \left( 1 - 3 \delta \right) \tilde X_m+ \frac{B}{1 - a/2},   \label{eq:normalshrink}
\end{align}
which implies that $Y_m \leq 1 - 3 \delta$, establishing \eqref{eq:Y_good}. 

\end{proof}

\subsection{Finer Quantization}
\label{sec:finer}
For $a\geq 2$, the controller receives an element of an $\lfloor a\rfloor+1$-element set instead of a single bit. In this case we restrict our attention to \emph{order-statistic} tests, meaning that we split the real line into $\lfloor a \rfloor +1$ intervals 
\begin{equation}
 (-\infty,w_{1,n}),[w_{1,n},w_{2,n}),\dots, [w_{\lfloor a\rfloor,n},\infty),
\end{equation}
and the controller receives the index $b_n\in\{0,1,\dots,\lfloor a\rfloor\}$ of the interval containing $X_n$. The only real issue is for the quantizer and the controller to agree upon a rule for \ver{updating} the values of $w_i$. However, this is easy; in the obvious generalization of our algorithm to higher $a$, \ver{ the (uniform) quantizer simply breaks up the interval $[-C_n,\,C_n]$ into $\lfloor a \rfloor+ 1$ equal parts, where $C_n$ is the same bound on the state magnitude as before. Both quantizer and controller follow the rules in \eqref{eq:Cmplusi} (with $a/2$ replaced by $a / (\lfloor a\rfloor+1)$ and in \eqref{eq:emergency} to update $C_n$. During the normal mode, the controller applies the control 
\begin{align}
 U_n = -C_n + C_n \frac {2 b_n +1 } {\lfloor a \rfloor +1} \label{eq:Uiboundeda}
\end{align}
which reduces to \eqref{eq:Uibounded} when $\lfloor a \rfloor = 1$.}

In the case $a <1$, the controller does nothing, which by Lemma~\ref{lemma:moments} achieves $\beta$-moment stability.

\section{Converse}
\label{sec:converse}
In this section, we prove the converse result in \thmref{thm:mainc} using information-theoretic arguments similar to those employed in \cite{nair2004stabilizability,kostina2016ratecost}. Then, we  use elementary probability to show an alternative converse result, which implies  \thmref{thm:mainc} unless $a$ is an integer.
\begin{proof}[Proof of \thmref{thm:mainc}]
Conditional entropy power is defined as
\begin{equation}
 N(X|U) \triangleq \frac 1 {2 \pi e} \exp\left( 2 h(X|U) \right)
\end{equation}
where $h(X|U) = -\int_{\mathbb R} f_{X, U}(x, u) \log f_{X|U = u}(x) dx$ is the conditional differential entropy of $X$. 

Conditional entropy power is bounded above in terms of moments  (e.g. \cite[Appendix 2]{zamir1992universal}): 
\begin{align}
N(X) &\leq \kappa_\beta \E{ |X |^{\beta}}^{\frac 2 \beta} \\
\kappa_\beta &\triangleq \frac 2 { \pi e} \left(e^{\frac 1 \beta} \Gamma \left( 1 + \frac 1 \beta \right) \beta^{\frac 1 \beta}\right)^2,
\end{align}
Thus, 
\begin{align}
\kappa_\beta \E{ |X_n |^{\beta}}^{\frac 2 \beta} &\geq N\left( X_n \right)\\
 &\geq N\left( X_{n} | U^{n-1}\right) \label{eq:bi},
\end{align}
where \eqref{eq:bi} holds because conditioning reduces entropy. Next, we show a recursion on $N\left( X_{n} | U^{n-1}\right)$:
\begin{align}
N\left( X_{n} | U^{n-1}\right) &= N(\mathsf A X_{n-1} + {Z}_{n-1} | U^{n-1}) \\
 &\geq a^2 N(X_{n-1} | U^{n-1} ) + N({Z}_{n-1})\label{eq:cb}\\
 &\geq a^2 N(X_{n-1} | U^{n-2} ) \exp \left( - 2 r  \right)  + N({Z}_{n-1}),  
 \label{eq:recur}
\end{align}
where \eqref{eq:cb} is due to the conditional entropy power inequality:\footnote{Conditional EPI follows by convexity from the unconditional EPI first  stated by Shannon \cite{shannon1948mathematical} and proved by Stam \cite{stam1959some}.}
\begin{align}
N(X + Y | U)\geq N(X|U) + N(Y|U) \label{eq:epi},
\end{align} 
which holds as long as $X$ and $Y$ are conditionally independent given $U$, and \eqref{eq:recur} is obtained by weakening the constraint $|U_{n-1}| \leq M$ to a mutual information constraint $I(X_{n-1}; U_{n-1} | U^{n-2}) \leq \log M = r$  
 and observing that 
\begin{align}
\min_{P_{U|X} \colon I(X; U) \leq r} h(X|U) &\geq h(X) - r. 
\end{align}  
It follows from \eqref{eq:recur} that $r > \log a$ is necessary to keep $N\left( X_{n} | U^{n-1}\right)$ bounded. Due to \eqref{eq:bi}, it is also necessary to keep $\beta$-th moment of $X_n$ bounded. 
\end{proof}

Consider the following notion of stability.\footnote{\ver{A more stringent to \defnref{def:weakstable} notion of stability in probability, in which $>0$ in the right side of \eqref{eq:stabprob} is replaced by $ = 1$, was considered in \cite[Def. 2.1]{matveev2008state}, and in \cite[Th. 3.1]{yuksel2016nonlinear}. }}
\begin{defn}
 The system is stabilizable in probability if there exists a control strategy such that for some bounded interval $\mathcal I$, 
\begin{equation}
 \limsup_{n \to \infty} \Prob{X_n \in \mathcal I} > 0. \label{eq:stabprob}
\end{equation}
\label{def:weakstable}
\end{defn}


As a simple consequence of Markov's inequality, if the system is moment-stable, it is also stable in probability. Therefore the following converse for stability in probability implies a converse for moment stability. 

\begin{thm}
Assume that $X_1$ has a density. To achieve stability in probability, $M \geq \lceil a \rceil$ is necessary. 
\label{thm:weak}
\end{thm}
\begin{proof}
We want to show that for any bounded interval $\mathcal I$, if $r < \log a$ then
\begin{equation}
 \limsup_{n \to \infty} \Prob{X_n \in \mathcal I} = 0.  \label{eq:noweakstability}
\end{equation}
At first we assume that the density of $X_1$ is bounded, that is, $|f_{X_1}(x)| \leq f_{\max}$ and that $X_1$ is supported on a finite interval, i.e. $|X_1| \leq x_{\max}$, for some constants $f_{\max}, x_{\max}$. 

 Since we are showing a converse (impossibility) result, we may relax the operational constraints by revealing the noises $Z_n,~n = 1, 2, \ldots$ noncausally to both encoder and decoder. Since then the controller can simply subtract the effect of the noise, we may put $Z_n \equiv 0$ in \eqref{eq:systemscalar}.  Then, $X_{n+1} = a^n X_1 + \tilde U_n$, where $\tilde U_n \triangleq \sum_{i = 0}^n a^{n - i} U_i$ is the combined effect of $t$ controls, which can take one of $M^n$ values, i.e. $U_n = u(m)$ if $X_1 \in \mathcal I_m$, $m = 1, \ldots, M^n$. Regardless of the particular choice of control actions $u(m)$ and quantization intervals $\mathcal I_m$, for any bounded interval $\mathcal I$,
\begin{align}
 \Prob{X_{n+1} \in \mathcal I} &= \Prob{a^n X_1 + \tilde U_n  \in \mathcal I} \label{eq:leba}\\
 &= \sum_{m = 1}^{M^n} \Prob{a^n X_1 + u(m) \in \mathcal I, X_1 \in \mathcal I_m} \\
 &\leq M^n a^{-n} f_{\max} |\mathcal I |, \label{eq:lebb}
\end{align}
and \eqref{eq:noweakstability} follows for any $M < a$, confirming the necessity of $M \geq \lceil a \rceil$ to achieve weak stability.\footnote{\ver{An argument similar to \eqref{eq:leba}--\eqref{eq:lebb} showing that the Lebesgue measure of $\mathcal I$ cannot be sustained is the key to the data-rate theorems for invariance entropy \cite{colonius2009invariance,colonius2013note,kawan2015network,dasilva2018robustness}.
}}

Finally, if the density of $X_1$ is unbounded,  consider the set  $\mathcal S_b \triangleq \left\{ x \in \mathbb R \colon f_{X_1}(x) \leq b\right\}$ and notice that since $1\{ f_{X_1}(x) > b\} \to 0$ pointwise as $b \to \infty$, by dominated convergence theorem,
\begin{align}
\Prob{X_1 \in \mathcal S_b} &= \int_{\mathbb R} f_{X_1}(x) 1\{ f_{X_1}(x) > b\} \mathrm d x 
\to 0 \text{ as } b \to \infty. 
\end{align}
Therefore for any $\epsilon > 0$, one can pick $b > 0$ such that $\Prob{X_1 \in \mathcal S_b} \leq \epsilon$. Then, since we already proved that \eqref{eq:noweakstability} holds for bounded $f_{X_1}$, we conclude
\begin{equation}
\limsup_{n \to \infty} \Prob{X_n \in \mathcal I} \leq  \epsilon + \limsup_{n \to \infty}  \Prob{X_n \in \mathcal I | X_1 \in \mathcal S_b}  = \epsilon,
\end{equation}
which implies that \eqref{eq:noweakstability} continues to hold for unbounded $f_{X_1}$.
 \end{proof}

The quantities $\lceil a \rceil$ and $\lfloor a\rfloor +1$ coincide unless $a$ is an integer, thus \thmref{thm:weak} shows that for non-integer $a$, the converse (impossibility) part of \thmref{thm:main} continues to hold in the sense of weak stability. Note that the proof of \thmref{thm:weak} relaxes the causality requirement. We conjecture that $\lceil a \rceil$ can be replaced by $\lfloor a \rfloor+1$ in its statement, but proving that will require bringing causality back in the picture, and the simple argument in the proof of \thmref{thm:weak} will not work. 

We conclude \secref{sec:converse} with a technical remark. 
\begin{remark}
The assumptions in \thmref{thm:weak} are weaker than those in \thmref{thm:mainc}, because the differential entropy of $X_1$ not  being $-\infty$ implies that $X_1$ must have a density. The  assumption that $X_1$ must have a density is not superficial. For example, consider $Z_i \equiv 0$ and $X_1$ uniformly distributed on the Cantor set, and $a = 2.9$. Clearly this system can be stabilized with 1 bit, by telling the controller at each step the undeleted third of the interval the state is at. This is lower than the result of \thmref{thm:main}, which states that $M^\star_\beta$ would be $3$ if $X_1$ had a density. Beyond distributions with densities, we conjecture that $M^\star_\beta$ will depend on the Hausdorff dimension of the probability measure of $X_1$.
\end{remark}

\section{Generalizations}
\label{section:generalizations}

In this section, we generalize our results in several directions. In most cases we only outline the mild differences in the proof.

\subsection{Constant-Length Time Delays}
\label{sec:constantdelay}
Many systems have a finite delay in feedback. To model this, we can force $U_n$ to depend on only the feedback up to round $n-\ell$, i.e. 
\begin{align}
 U_n = \mathsf g_n( \mathsf f_1(X_1), \mathsf f_2(X^2), \ldots, \mathsf f_{n-\ell} (X^{n-\ell})),
\end{align}
where $\mathsf f_n(X^n)$ is the quantizer's output at time $n$, as before.  

We argue here that this makes no difference in terms of the minimum number of bits required for stability. We state the modified result next.

\begin{thm}

\label{thm:main3}

Let $X_1$, $Z_n$ in \eqref{eq:systemscalar} be independent random variables with bounded $\alpha$-moments. Assume that $h(X_1) > - \infty$.
The minimum number of quantization points to achieve $\beta$-moment stability, for any $0 < \beta < \alpha$ and with any constant delay $\ell$ is given by $\lfloor a\rfloor + 1$. 
\end{thm}

\begin{proof}

The problem here is that the encoder sees the system before the controller can act on it. However, if we also delay the encoder seeing the system by $\ell$ time steps, then we can directly use the algorithm we have already constructed. Specifically, if our artificially delayed sequence of system states is $\{\tilde X_n\}$, then the real sequence is given by 
\begin{equation}
 X_n=a^{\ell}\tilde X_n+a^{\ell-1}Z_{n+1}+...+Z_{n+\ell}.
\end{equation}

By \thmref{thm:main}, we can keep $\mathbb E[|\tilde X_n|^{\beta}]$ bounded for any $\beta<\alpha$ \ver{by applying to $\tilde X_n$ the control in \eqref{eq:Uiboundeda} in normal mode and $U_n =0$ in the emergency mode. Since with delay, the controller acts on $X_n$ rather than on $\tilde X_n$, we multiply the control action   \eqref{eq:Uiboundeda} by $a^\ell$ to achieve the same effect on $\tilde X_n$ as without delay.  Furthermore,} each $Z_i$ has bounded $\beta$ moment, so by \lemref{lemma:moments} their sum will have bounded $\beta$-moment, as desired. 

The converse is obvious as even with $\ell = 0$, \thmref{thm:mainc} asserts that the system cannot be stabilized with fewer than $\lfloor a\rfloor + 1$ bits. 

\end{proof}

\subsection{Packet drops}

Suppose that the encoder cannot send information to the controller at all time-steps. Instead, the encoder can only send information at a deterministic set $\mathcal T\subseteq\mathbb N$ of times. 
Formally, 
\begin{align}
 U_n = \mathsf g_n( \{\mathsf f_n(X^n)  \colon n \in \mathcal T\}).
\end{align}

As long as the density of $\mathcal T$ is high enough on all large, constant-sized scales, the same results go through.

\begin{defn}
A set $\mathcal T\subseteq\mathbb N$ is \emph{strongly $p$-dense} if there exists $N$ such that for all $n$ we have 
\begin{equation}
 \frac{| n + i \colon n + i \in \mathcal T, ~ i = 0, \ldots, N - 1 | }{N} > p.
\end{equation}
\label{defn:pdense}
\end{defn}
Note that the constant delay scenario in \secref{sec:constantdelay} amounts to control on a strongly $p$-dense set, with $p \in [0, 1)$ as close to $1$ as desired.
\begin{thm}

\label{thm:main4}

Let $X_1$, $Z_n$ in \eqref{eq:systemscalar} be independent random variables with bounded $\alpha$-moments. Assume that $h(X_1) > - \infty$. The minimum number of quantization points to achieve $\beta$-moment stability is $\lfloor a \rfloor + 1$, for any $0 < \beta < \alpha$ and on any strongly $p$-dense set with some $p\in [0,1]$ large enough so that 
\begin{equation}
 \left(\lfloor a\rfloor+1\right)^p>a. \label{eq:pdenselargep}
\end{equation}
\end{thm}

\begin{proof}
The requirement \eqref{eq:pdenselargep} ensures that the bounded case works; indeed, it is equivalent to
\begin{equation}
 \left(\frac{a}{\lfloor a\rfloor+1}\right)^pa^{1-p}<1, 
\end{equation}
which means that the logarithm of the range of $X_n$ decreases on average each time-step.

In the unbounded noise case, we perform the same basic algorithm, but ensuring that the normal mode has enough times in $\mathcal T$, so the duration of the normal mode gets longer if $N$ is large. Likewise, in the emergency mode, it will take longer to catch blow-ups. However, a much weaker condition on $\mathcal T$ suffices for the emergency mode to end: even if $\mathcal T$ contains only $1$ element out of every $N$, we make the probe factor $P$ large enough depending on $N$. The difference between probing every $N$ time steps vs. every time step at most a factor of $a^N$ which is a constant.
\end{proof}

\subsection{Dependent Noise}
\label{sec:dependent}
Here we address a modification in which the noise is correlated rather than independent. 

\begin{prop}

Suppose $\{Z_n\}$ is a Gaussian process whose covariance matrix $\mat M$ (for any number of samples) has spectrum bounded by $\lambda$. Then there is an independent Gaussian process $\{Z'_n\}$ such that the random variables $\{Z_n+Z'_n\}$ are i.i.d. Gaussians with variances $\lambda^2$.
\label{prop:correlated}
\end{prop}

\begin{proof}
Just make the covariances matrices of $\{Z_n\}$ and $\{Z'_n\}$ add to $\lambda I$; the assumption means both are positive semidefinite, hence define Gaussian processes.
\end{proof}

If the rows of $\mat M$ have $\ell^1$ norm at most $\lambda$, the assumption of \propref{prop:correlated} that $\mat M$ has spectrum bounded by $\lambda$ will be satisfied.  Indeed, we can add a positive semidefinite matrix to such an $\mat M$ to obtain a diagonal matrix with each entry at most $\lambda$: if $\mat M$ has an entry of $x$ at positions $(i,j)$ and $(j,i)$ then we add $|x|$ to the $(i,i)$ and $(j,j)$ entries; it is easy to see that the symmetric matrix 
\begin{equation}
 \begin{pmatrix} |x|& -x \\ -x& |x|\end{pmatrix}
\end{equation}
is always positive semidefinite. Therefore doing this for all non-diagonal entries adds a positive semidefinite matrix to $\mat M$ and still results in a spectrum contained in $[0,\lambda]$, meaning $\mat M$ also had spectrum contained in $[0,\lambda]$.

\begin{thm}

The results in Theorems \ref{thm:main}, \ref{thm:main3}, \ref{thm:main4} extend to the case when $\{Z_n\}$ is correlated Gaussian noise whose covariance matrix has bounded spectrum.
\label{thm:dependent}
\end{thm}

\begin{proof}

As a result of \propref{prop:correlated}, for any Gaussian noise with known covariance matrix $\mat M$ of bounded spectrum, the controller can simply add extra noise to the system via $U_n$ to effectively make the noise i.i.d. Gaussian, reducing this scenario to the i.i.d. case.

\end{proof}

\subsection{Vector systems}
\label{sec:vector}

The results generalize to higher dimensional systems
\begin{equation}
 X_{n+1}=\mat A X_n+ Z_n-\mathsf B U_n, \label{eq:vector2}
\end{equation}
where $\mat A$ is a $d\times d$ matrix and $Z_n,~ U_n$ are vectors. The dimensionality of control signals $U_n$ can be less than $d$, in which case $\mat B$ is a tall matrix.

For the controls to potentially span the whole space $\mathbb R^d$ when combined with the multiplication-by-$\mathsf A$ amplification, the range of 
\begin{equation}
\left[ \mat B,\, \mat A \mat B,\, \mat A^2 \mat B,\,\dots,\, \mat A^{d-1} \mat B \right] \label{eq:controllability}
\end{equation}
needs to span $\mathbb R^d$. Such a pair $(\mathsf A,\mathsf B)$ of matrices is commonly referred to as \emph{controllable}. (The $0, \ldots, d-1$ powers of $\mat A$ are sufficient in \eqref{eq:controllability}, because by the Cayley-Hamilton theorem any higher power of $\mat A$ is a linear combination of those lower powers). For our results to hold, a weaker condition suffices, namely, we need {\it stabilizability} of $(\mat A, \mat B)$, which is to say that only unstable modes need be controllable. More precisely, in the {\it canonical} representation  of a linear system,
\begin{align}
\begin{bmatrix}
 X^u_{n+1} \\
 X^s_{n+1}
\end{bmatrix} 
= 
\begin{bmatrix}
 \mat A^u & \mat A^\prime \\
 \mat 0 & \mat A^s
\end{bmatrix} 
\begin{bmatrix}
 X^u_n \\
 X^s_n
\end{bmatrix} 
+ Z_n - 
\begin{bmatrix}
\mat B^u \\
 \mat 0
\end{bmatrix} 
U_n,
\label{eq:canonical}
\end{align}
where the matrix $\mat A^s$ has all stable eigenvalues,  the state coordinates $X^s_{n+1}$ cannot be reached by the control $U_n$. Stabilizability means that the pair $(\mathsf A^u, \mathsf B^u)$ is controllable, which ensures that unstable modes can be controlled.

The idea behind our generalization to the vector case, previously explored in e.g. \cite{nair2004stabilizability}, is that we can decompose $\mathbb R^d$ into eigenspaces of $\mat A$ and rotate attention between these parts.

\begin{thm}

\label{thm:main5}

Consider the stochastic vector linear system in \eqref{eq:vector2} with $(\mat A, \mat B)$ stabilizable. Let $X_1, ~ Z_n$ be independent random $\mathbb R^d$-valued random vectors with bounded $\alpha$-moments. Assume that   $h(X_1) > - \infty$.  Let $(\lambda_1,...,\lambda_d)$ be the eigenvalues of $\mat A$, and set
\begin{equation}
 a \triangleq \prod_{j=1}^d \max(1,|\lambda_j|). \label{eq:avec}
\end{equation} 
Then for any $0<\beta<\alpha$, the minimum number of quantization points to achieve $\beta$-moment stability is
\begin{equation}
M_\beta^\star = \lfloor a\rfloor + 1. \label{eq:Mvec}
\end{equation}

\end{thm}

\begin{proof}
We first consider the case $U_{n} \in \mathbb R^d$, $\mat B = \mat I$:
\begin{equation}
 X_{n+1}=\mat A X_n+ Z_n - U_n, \label{eq:vector}
\end{equation}
and then explain how to deal with the general stabilizable system in \eqref{eq:vector2}.

By using a real Jordan decomposition, we can block-diagonalize $\mat A$ into 
\begin{equation}
 \mat A = \bigoplus_j \mat A_j
\end{equation}
where $\mat A_j \colon \mathcal V_j \to \mathcal V_j$ with 
\begin{equation}
 \bigoplus_j \mathcal V_j=\mathbb R^d
\end{equation}
such that:

\begin{enumerate}

\item The spectrum of each $\mat A_j$ is either a single real $\lambda_j$ (possibly with multiplicity) or a pair of complex numbers $\lambda_j,\bar{\lambda}_j$ (with equal multiplicity). 

\item The spectral norm of $\mat A_j^k$ is $|\lambda_j|^k k^{O(1)}$.

\end{enumerate}

This decomposition splits any vector $X_n$ into a sum $X_n=\sum_j X_{n,j}$. Each $X_{n,j}$ individually satisfies a control equation with matrix $\mat A_j$, and we will control these separately. Indeed, if 
\begin{equation}
\sup_n\mathbb E[|X_{n,j}|^{\beta}]<\infty 
\label{eq:Xnjbound}
\end{equation}
for all $j$, then we get the desired result. (Note that we do \emph{not} need to assume that the noise $Z_n$ behaves independently on each subspace; getting from separate moment bounds on $X_{n,j}$ to a moment bound on $X_n$ does not require any sort of independence.)

If $|\lambda_j|<1$, we can leave that subspace alone; we will have \eqref{eq:Xnjbound} without doing anything (by Lemma~\ref{lemma:moments}). If $\lambda_j\geq 1$ we will act on this subspace at times in $\mathcal T_j\subseteq\mathbb N$ where $\mathcal T_j$ is strongly $p_j$-dense for some $p_j$ (\defnref{defn:pdense}) with 
\begin{equation}
 (\lfloor a\rfloor + 1)^{p_j}>|\lambda_j|^{\dim(\mathcal V_j)}.
\end{equation}
The assumption \eqref{eq:avec} precisely means that we can pick such $p_j$ with $\sum_j p_j<1$. Generating a partition of $\mathbb N$ into strongly $p_j$-dense sets $\mathcal T_j$ is simple if this constraint holds.

Now we are left to explain how to handle the problem on each $\mathcal V_j$ separately. If $\dim(\mathcal V_j)=1$, then we have done this before. The key point is the second property above on the growth of the spectral norm of $\mat A_j^k$; for fixed, large enough $k$, this growth is slow enough that we can use the same procedure, and the non-trivial Jordan blocks won't matter. 

Now we proceed as before in rounds of $k$ steps, except that the encoder sends everything at the end of a round rather than bit-by-bit (we can do this by introducing a constant amount of delay). At the end of $k$ steps, assuming that $X_{n,\,j}$ started in some ball $B_C(0)$ at the start of the round for $C$ large, the ending value $X_{n+k,\, j}$ will with high probability be contained in a ball $B_{C'}(0)$ with $C'$ given by (for some $\epsilon > 0$)
\begin{equation}
 C'=(|\lambda_j|+\varepsilon)^{\left(\frac{ k p_j}{\dim \mathcal V_j}\right)}C
\end{equation}
assuming that no noise term was very large (the ``high probability" is independent of $j$ by another use of Lemma~\ref{lemma:moments} - note that the high-dimensionality doesn't matter since the Euclidean norm is subadditive).

We also recall that any set in $\mathcal S_j$ of diameter $D$ can be covered by at most $O(\left(\frac{D}{r}\right)^{\dim \mathcal S_j})$ balls of radius $r$. Hence for large enough $k$, we can cover $B_{C'}(0)$
with $O\left(\left(\lfloor a\rfloor + 1 -\varepsilon \right)^{k\,p_j}\right)$ balls of radius $\left(\frac{(|\lambda_j|+\varepsilon)C}{\lfloor a\rfloor + 1 -\varepsilon}\right)$.

The upshot of this is that at the end of a round of length $k$, assuming no blow-up happened, the encoder has enough bandwidth to point to one of many balls of smaller radius than the starting ball and assert that $X_{n,j}$ is now inside that ball. Hence, typical behavior of the system will reduce the radius of $X_{n,j}$ by a constant factor each round.

Emergency mode proceeds in the same way as before, using balls of larger and larger size. The effect is still that the $\beta$-moment of the radius decreases in expectation when large, hence is bounded.

To prove the converse, we can project out stable eigenmodes of $\mat A$ as done in \cite{kostina2016ratecost}, and then apply a straightforward generalization of the reasoning in \secref{sec:converse} to the resulting vector system. This converse will apply to the system with low-dimensional controls in \eqref{eq:vector2}, because we can always augment the matrix $\mat B$ to make it full rank and extend the dimension of the control signal accordingly.

To show an achievable scheme for \eqref{eq:vector2}, we will reduce the problem to the delayed version of \eqref{eq:vector}. Although we only addressed delays in the $1$-dimensional setting in \secref{sec:constantdelay}, the exact same argument shows that delays change absolutely nothing in all dimensions. We will focus on the case of controllable $(\mathsf A, \mathsf B)$, because if $(\mathsf A, \mathsf B)$ is merely stabilizable we can always ignore the uncontrollable stable part as per the canonical representation \eqref{eq:canonical}.  We will use the spanning set of matrices to give an arbitrary control with a delay of $\ell$ steps, where $\ell \leq d - 1$ is such tha t
the range of 
\begin{equation}
\left[ \mat B,\, \mat A \mat B,\, \mat A^2 \mat B,\,\dots,\, \mat A^{\ell} \mat B \right] 
\end{equation}
spans $\mathbb R^d$.
Then any vector $v\in\mathbb R^d$ can be written as
\begin{equation}
 v=\mat B v_0+\mat A \mat B v_1+\dots +\mat A^\ell \mat B v_\ell, \label{eq:decomp}
\end{equation}
where $v_i\in \mathrm{ran}(\mat A^i \mat B)$ for each $i$, and $|v_i|=O(|v|)$. 

Now, suppose that the sequence $\{\hat U_n\}_{n \in\mathbb Z^+}$ solves the control problem with delay $\ell$, meaning that the control $\hat U_n$ is chosen at time $n$ but kicks in at time $n+\ell$. That is, $\{\hat U_n\}$ is chosen so that the sequence $\hat X_n$ given by
\begin{equation}
 \hat X_{n+1}= \mat A \hat X_n +Z_n - \hat U_{n - \ell} \label{eq:vecdelay}
\end{equation}
has bounded moments. We assume that the noises $Z_n$ are common to \eqref{eq:vector2} and \eqref{eq:vecdelay}, so that $\{X_n\}$ and $\{\hat X_n\}$ are coupled together rather than independent.
Per \eqref{eq:decomp} we can write
\begin{equation}
 \hat U_n = \mat B\, U_{0,n}+ \mat A \mat B\, U_{1,n} + \mat A^2 \mat B\, U_{2,n} + \dots + \mat A^\ell \mat B\, U_{\ell, n}.
\end{equation}
To realize control $\hat U_n$ that takes full effect at time $n+\ell$, we can have $\hat U_n$ contribute $\mat B\, U_{i,n}$ to the control at time $n +\ell -i$. If we do this for every $n$, however, we see that the control $\tilde U_n = \mat B\, U_n$ applied at time $n$ will consist of contributions from all $\hat U_{n-\ell}, \ldots, \hat U_{n}$:
\begin{equation}
 U_n=\sum_{i=0}^\ell U_{i,n-\ell+i}, \label{eq:controloriginal}
\end{equation}
and the actual accumulated control by time $n+\ell$ is larger than $\hat U_n$.  Therefore, by applying \eqref{eq:controloriginal} to the state of the original system $X_n$  \eqref{eq:vector2} at time $n$, we will not get exactly $\hat X_n$. However, the difference $\hat X_n - X_n$ is a finite sum of terms of type $\mat A^{j_1}\, U_{j_2,\ell}$ that are bounded, according to \eqref{eq:Uibounded}, in terms of the same majorizing sequence $\tilde X_n$ (\eqref{eq:Xtilde}, \lemref{lemma:max}) that we used to bound the $\beta$-moment of $\hat X_n$. Since  $\tilde X_n$ has bounded $\beta$-moment according to \eqref{eq:momentdecrease} (although we haven't emphasized it, $\tilde X_n$ is also bounded in $\beta$-moment in higher dimensions just as in the $1$-dimensional case), we conclude that the finite number of extra controls has no bearing on stabilizability.  
\end{proof}

\section{Conclusion}
This paper studies the minimum number of bits necessary and sufficient for stability, when fixed-rate quantizers are used, and proves that conveying $\lfloor a\rfloor + 1$ distinct values is necessary and sufficient to achieve $\beta$-moment stability, where $a$ is defined in \eqref{eq:avec}, provided that the independent additive noises have bounded $\alpha$ moments, for some $\alpha > \beta$.   \thmref{thm:main}, which is the main technical result of the paper, proposes and analyses a time-varying strategy to achieve stability of a scalar system under this minimum communication requirement. We use probabilistic arguments to show this result.  \thmref{thm:mainc} shows a matching converse (impossibility) result, attesting that no strategy can achieve stability with a lower amount of communication. We use information-theoretic arguments to show this result. \thmref{thm:weak} relaxes the assumptions of \thmref{thm:mainc}, and shows, using a purely probabilistic argument, that $\lceil a\rceil$ distinct messages are necessary for stability even in the absence of additive noise. Generalizations to constant-length time delays, communication channels with packet drops, dependent noise, and vector systems are presented in \thmref{thm:main3},  \thmref{thm:main4}, \thmref{thm:dependent},  \thmref{thm:main5}, respectively. 

In \cite{ranade2018unboundedArxiv}, we applied a similar strategy to stabilize a system with random gain $a$ (which is constant in the present paper) using finitely many bits at each time step. 

An advantage of the scheme presented in this paper compared to \cite{kostina2016ratecost} is that it uses a fixed number of bits at each time step, and thus is directly compatible with standard block error-correcting codes used for the transmission over noisy channels. Analyzing how our strategy can be applied together with an appropriate error-correcting code to control over noisy channels and whether fundamental limits can be attained that way is an important future research direction.  

While we picked the constants to guarantee a bounded $\beta$-moment,  we did not try to optimize them in order to minimize it.  A natural future research direction, then, is to study, in the spirit of \cite{kostina2016ratecost}, the tradeoff between rate and the attainable $\beta$-moment. It will be interesting to see whether our scheme can approach the lower bound in \cite{kostina2016ratecost}, and to compare its performance with that of the Lloyd-Max quantizer, explored in the context of control in \cite{khina2017lloydmax}.

\appendix

\subsection{Proof of \eqref{eq:normalbound}}
\label{sec:proofnormalbound}
For $1 \leq i \leq k$, we express the system state at time $m + k$ in terms of the system state at time $m+i$: 
 \begin{align} 
X_{m+k} =
 a^{k-i} \Bigg( &X_{m+i} + \sum_{\ell = 0}^{k-i-1} a^{- \ell - 1} U_{m+i+ \ell}
 \dcol{\notag\\ &}{}
+ \sum_{\ell = 0}^{k-i-1} a^{- \ell - 1 } Z_{m +i+ \ell} \Bigg). \label{eq:Xkicontrols}
\end{align}

Applying \eqref{eq:Uibounded}, \eqref{eq:Cmplus1} and \eqref{eq:Cmplusirec}, we can crudely bound the cumulative effect of controls on $X_{m+k}$ as
 \begin{align}
\dcol{\!\!\!\!}{}  \left| \sum_{\ell = 0}^{k-1} a^{ - \ell - 1} U_{m+i + \ell} \right|
&\leq 
\left( a/2 \right)  \sum_{\ell = 1}^{\infty} a^{ - \ell - 1} 
\dcol{ \notag\\ \phantom{=}&}{}
\left( \left( a/2 \right)^{\ell} C_{m+i} +   \frac{ 1 - ({ a}/2)^{\ell}} {1 - { a}/2} B \right) \\ 
&= 
C_{m + i}  + \frac B { a-1 } \\
&\leq \left( a/2 \right)^{-k}  C_m + \frac{a B}{(2 - a) (a - 1)}
\label{eq:controlkinormal} 
\end{align}
Unifying \eqref{eq:Xkicontrols} and \eqref{eq:controlkinormal}, we get
\begin{align}
|X_{m+i}| &\leq |X_{m+k}|  + \left( a/2 \right)^{-k}  C_m 
\dcol{\notag\\ \phantom{=}&}{}
+ \frac{a B}{(2 - a) (a - 1)} + \sum_{\ell = 0}^{k-i-1} a^{- \ell - 1 } |Z_{m +i+ \ell} | \label{eq:boundXik}
\end{align}
By \lemref{lemma:moments}, the sum of random variables on the right ride of \eqref{eq:boundXik} has uniformly bounded $\alpha$-moments, and since by definition of $\tilde X_n$ in \eqref{eq:Xtilde}, $\tilde X_m \leq C_m$ and $|X_{m+k}| \leq \tilde X_{m+k}$, \eqref{eq:normalbound} follows by the means of \eqref{eq:calphaeps}.

\subsection{Proof of \lemref{lemma:max}}
Combining \eqref{eq:Xmj}, \eqref{eq:controlnormal} and $|X_m| \leq C_m$ yields for $ i = 1, 2, \ldots, k+\tau$, 
\begin{align}
 |X_{m+i}| \leq a^{i} \left( 2 C_m + \frac{B}{a-1} 
+ \sum_{\ell = 0}^{i-1} a^{- \ell - 1 } |Z_{m + \ell}| \right),  \label{eq:xboundnormal}
\end{align}
Maximizing the right side of \eqref{eq:xboundnormal} over $1 \leq i \leq k + j$ and using \eqref{eq:xboundnormal}, 
we conclude
\begin{align}
 \max_{1 \leq i \leq k + j} |X_{m+i} | 
 &\leq a^{k+j} 
 \bigg( 2 C_m + \frac{B}{a-1} 
 \dcol{\notag\\&}{}
+ \sum_{\ell = 0}^{k+j-1} a^{- \ell - 1 } |Z_{m + \ell}| \bigg)  \label{eq:mostmax}, 
\end{align}
It remains to bound $C_{m+k+j} $. If $j = 0$, we may simply apply \eqref{eq:Xmpluskbound}, which means, crudely, 
\begin{equation}
 C_{m+k+j} \leq \text{right side of \eqref{eq:mostmax}} + \frac{ a^{k+j} B}{1 - a/2}. \label{eq:Cendbound0}
\end{equation}
If $j > 0$, since the round did not end on step $m + k + j - 1$, we have $C_{m+k+j - 1} < |X_{m+k+j - 1}|$, which means that
\begin{equation}
C_{m+k+j} < P |X_{m+k+j - 1}|. \label{eq:Cendbound}
\end{equation}
Combining \eqref{eq:mostmax}, \eqref{eq:Cendbound0} and \eqref{eq:Cendbound} yields \eqref{eq:max}.

\bibliographystyle{IEEEtran}
\bibliography{../../Bibliography/control,../../Bibliography/it,../../Bibliography/vk}

\begin{IEEEbiography}
    [{\includegraphics[width=1in,height=1.25in,clip,keepaspectratio]{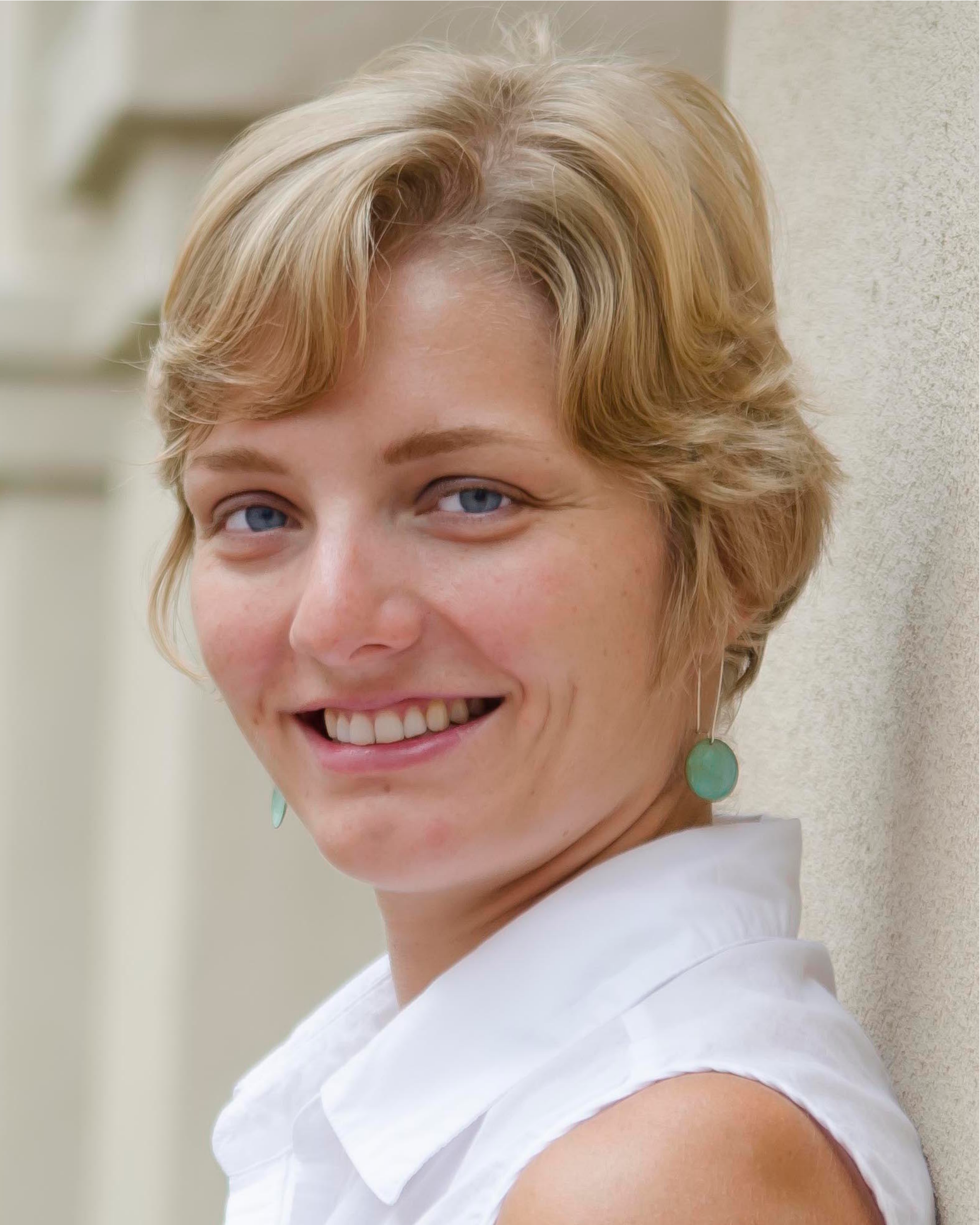}}]{Victoria Kostina}
received the bachelor's degree from Moscow Institute of Physics and Technology (MIPT) in 2004, the master's degree from University of Ottawa in 2006, and the Ph.D. degree from Princeton University in 2013.  During her studies at MIPT, she was affiliated with the Institute for Information Transmission Problems of the Russian Academy of Sciences. 
She is currently a Professor of electrical engineering and computing and mathematical sciences at California Institute of Technology. Her research interests include information theory, coding, control, learning, and communications. 
 She received the Natural Sciences and Engineering Research Council of Canada postgraduate scholarship during 2009--2012, Princeton Electrical Engineering Best Dissertation Award in 2013, Simons-Berkeley research fellowship in 2015 and the NSF CAREER award in 2017.
\end{IEEEbiography}

\begin{IEEEbiography}
    [{\includegraphics[width=1in,height=1.25in,clip,keepaspectratio]{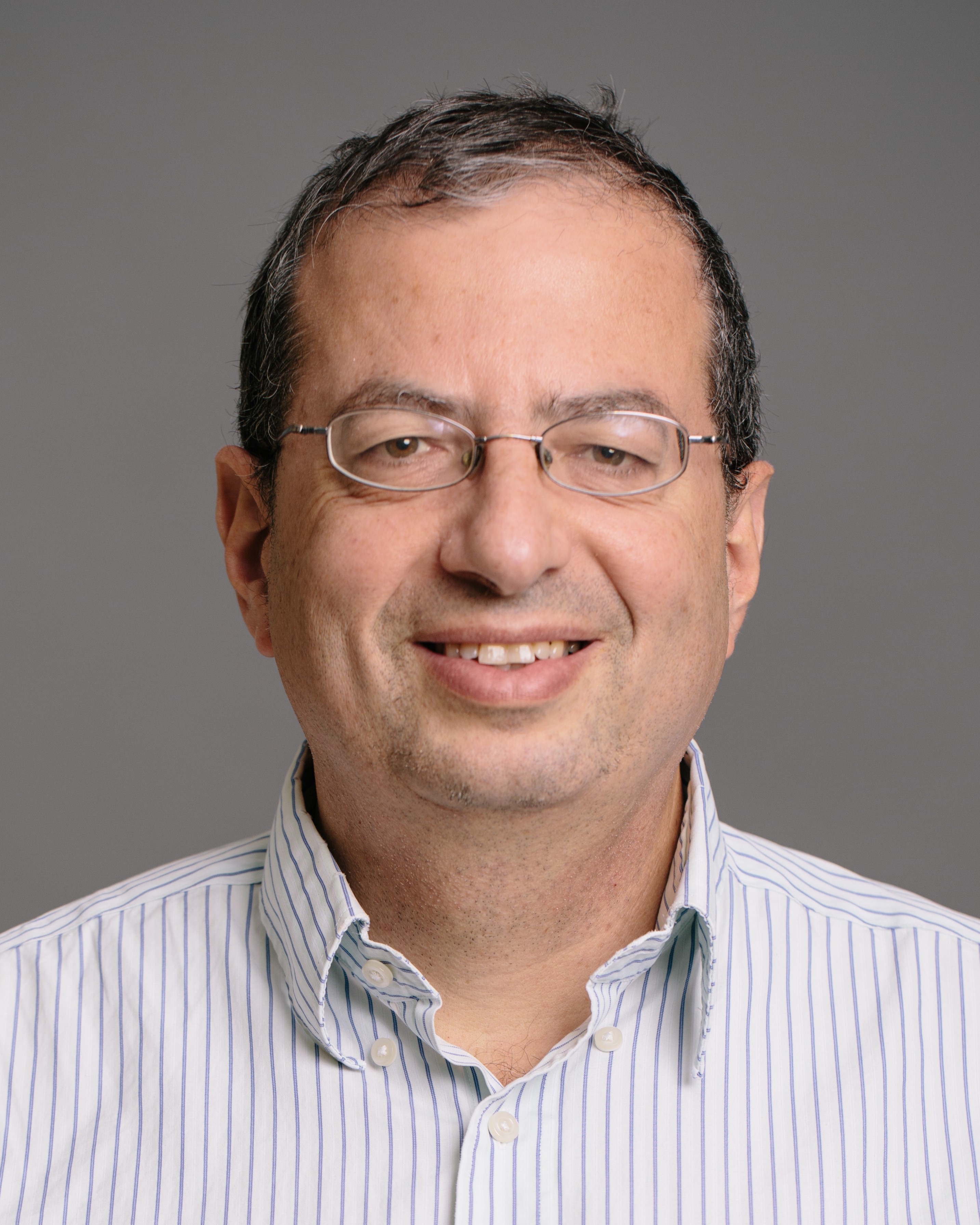}}]{Yuval Peres}
obtained his PhD  in 1990 from the Hebrew University in Jerusalem. In 1993, he joined the faculty of the University of California at Berkeley, where he served as a professor in the mathematics and statistics departments until 2006. From 2006 to 2018, he was a Principal Researcher at Microsoft Research.  He has also taught at Yale and at the Hebrew University, and is currently visiting Kent State University.  Yuval Peres has published more than 300 papers with 200 co-authors and has mentored 21 PhD theses. His research encompasses most areas of probability theory, including random walks, Brownian motion, percolation, and random graphs. He has co-authored books on Markov chains, probability on graphs, game theory and  Brownian motion. Dr. Peres is an IMS fellow and a recipient of the Rollo Davidson prize and the Loeve prize. In 2002, he was an invited speaker at the International Congress of Mathematicians in Beijing, and in 2016 he was elected to the National Academy of Sciences.
\end{IEEEbiography}

\begin{IEEEbiography}
    [{\includegraphics[width=1in,height=1.25in,clip,keepaspectratio]{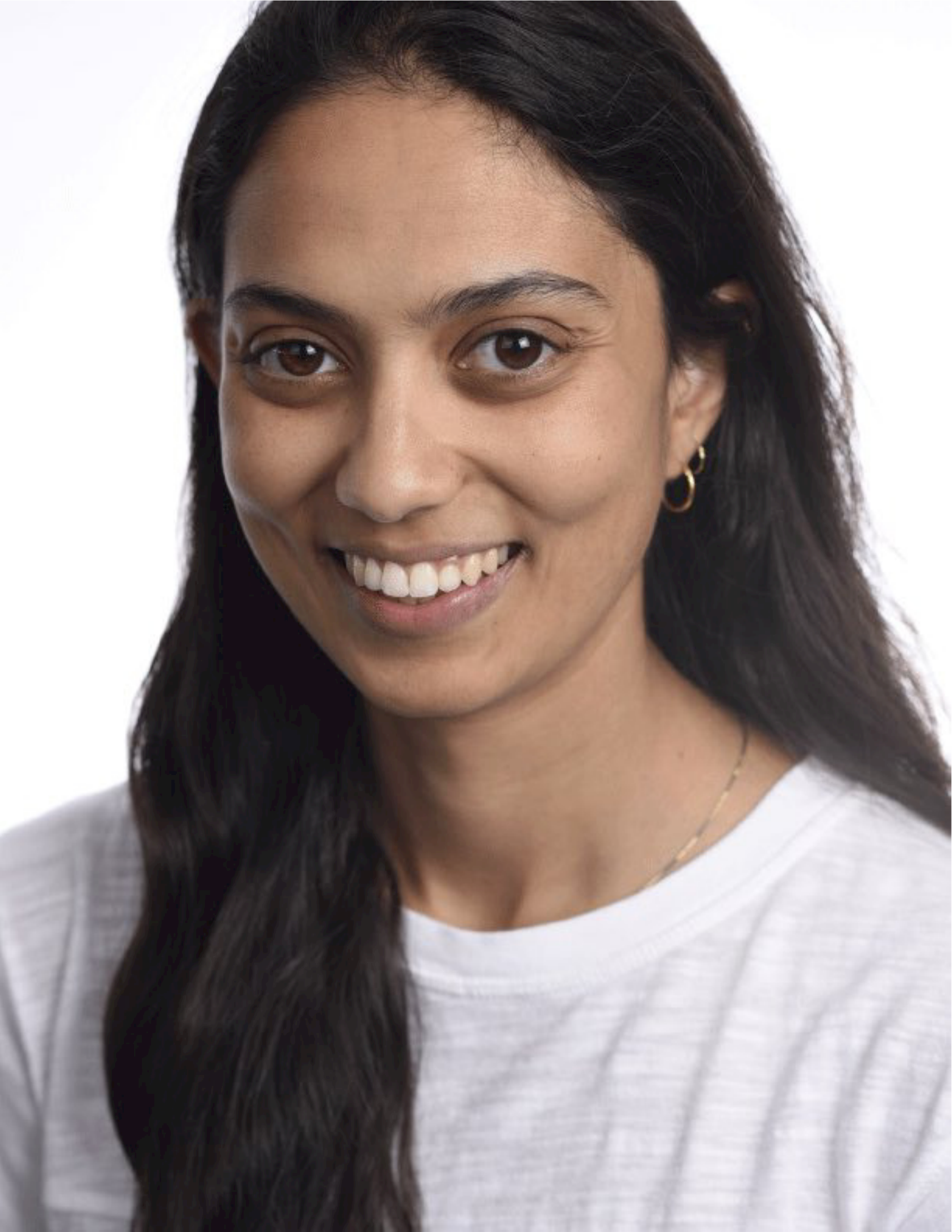}}]{Gireeja Ranade}
is an Assistant Teaching Professor at UC Berkeley. Before this she was a Researcher at Microsoft Research AI, Redmond. She received an MS and PhD in EECS from UC Berkeley and an SB in EECS from MIT. Her research interests have revolved around understanding stochastic systems and problems in control theory, information theory and wireless communications.
\end{IEEEbiography}

\begin{IEEEbiography}
    [{\includegraphics[width=1in,height=1.25in,clip,keepaspectratio]{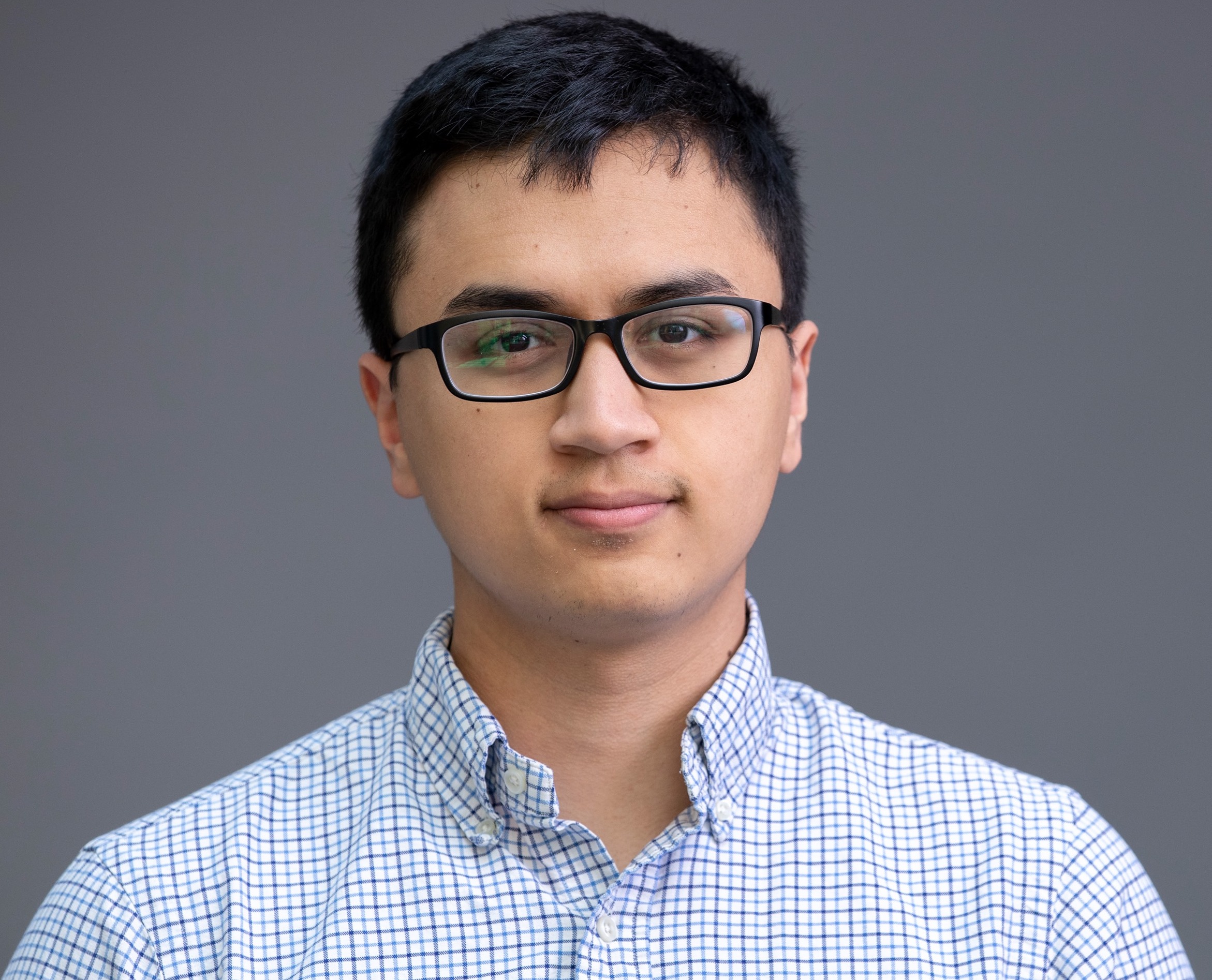}}]{Mark Sellke}
 is a PhD student in mathematics at Stanford advised by S\'ebastien Bubeck and Andrea Montanari. He graduated from MIT in 2017 with a B.S. in mathematics and from the University of Cambridge with a Masters in mathematics with distinction in 2018. He is the recipient of an NSF Graduate Fellowship and a Stanford Graduate Fellowship. Mark's primary research interests are in probability and theoretical machine learning. His work "Chasing Convex Bodies Optimally" won the best paper and best student paper awards at SODA~2020.
\end{IEEEbiography}

\end{document}